\documentclass[12pt]{article}

\usepackage[left=1.25in, top=1in]{geometry}
\usepackage{amsthm}
\usepackage{amsmath}
\usepackage{mathtools}
\usepackage{epsfig}
\usepackage{color}
\usepackage{verbatim}
\usepackage{amsfonts}
\usepackage{url}
\usepackage[small,bf]{caption}
\usepackage{float}
\usepackage[square,comma,numbers,sort]{natbib}
\usepackage{multirow}

\usepackage{titlesec}
\titlespacing*{\section}{0pt}{0.6\baselineskip}{0.6\baselineskip}
\titlespacing*{\subsection}{0pt}{0.6\baselineskip}{0.6\baselineskip}

\usepackage{url}
\usepackage{algorithmicx}
\usepackage{algpseudocode}
\usepackage[dvipsnames]{xcolor}
\usepackage{wrapfig,lipsum,booktabs}

\newcommand{\myparagraph}[1]{\vspace*{-1.9ex}\paragraph*{\normalsize\bf{#1}}}

\newlength{\figurewidth}
\newlength{\smallfigurewidth}
\algtext*{EndFunction}% Remove "end function" text
\algtext*{EndProcedure}% Remove "end procedure" text
\algtext*{EndFor}% Remove "end for" text
\algtext*{EndWhile}% Remove "end while" text
\algtext*{EndIf}% Remove "end if" text
\algloopdefx{NIf}[1]{\textbf{if} #1 \textbf{then}}
\algloopdefx{NElse}{\textbf{else}}
\algloopdefx{NElseIf}{\textbf{else if}}
\algloopdefx{NForAll}[1]{\textbf{for each} #1 \textbf{do}}
\algloopdefx{NWhile}[1]{\textbf{while} #1 \textbf{do}}
\algloopdefx{NFor}[1]{\textbf{for} #1 \textbf{do}}

\newtheorem{lemma}{Lemma}

\newtheorem{theorem}{Theorem}

\theoremstyle{definition}

\newcommand{\rank}{\mathsf{rank}}
\newcommand{\select}{\mathsf{select}}
\newcommand{\access}{\mathsf{access}}
\newcommand{\pred}{\mathsf{pred}}

\begin{document}

\title
{\vspace{-1.0cm}{\Large Run Compressed Rank/Select for Large
    Alphabets}\thanks{This work is supported by the Academy of Finland via grant
    1294143 and by the EU grant H2020-MSCA-RISE-2015 BIRDS690 No.\ 690941. The first author received funding from Fondecyt grant 3170534 and Basal Funds FB0001, Conicyt, Chile. Part of this work was developed during the Shonan Meeting 126 ``Computation over Compressed Structured Data''.}}

%\author{%
%Jos{\'e} Fuentes-Sep{\'u}lveda$^{\ast}$, Juha K{\"a}rkk{\"a}inen$^{\dag}$, Dmitry Kosolobov$^{\dag}$, Simon J. Puglisi$^{\dag}$\\[0.5em]
%{\small\begin{minipage}{\linewidth}\begin{center}
%\begin{tabular}{ccc}
%$^{\ast}$Department of Computer Science, & \hspace*{0.5in} & $^{\dag}$Department of Computer Science,\\
% University of Chile & \hspace*{0.5in} & University of Helsinki \\
%%Street Address One && Street Address Two \\
%Santiago, Chile && Helsinki, Finland\\
%\url{jfuentess@dcc.uchile.cl} && \url{juha.karkkainen@cs.helsinki.fi}\\
% && \url{dkosolobov@mail.ru} \\
% && \url{puglisi@cs.helsinki.fi}
%\end{tabular}
%\end{center}\end{minipage}}
%}

\author{\normalsize Jos{\'e} Fuentes-Sep{\'u}lveda$^1$, Juha K{\"a}rkk{\"a}inen$^2$, Dmitry Kosolobov$^2$, and Simon J.\ Puglisi$^2$\\[1ex]
\footnotesize $^1$ Department of Computer Science, University of Chile, Santiago, Chile\\[-0.8ex]
\footnotesize $^2$ Helsinki Institute for Information Technology,\\[-0.8ex]
\footnotesize Department of Computer Science, University of Helsinki, Helsinki, Finland\\[-3.8ex]}

\date{}
\maketitle
\thispagestyle{empty}

\begin{abstract}
Given a string of length $n$ that is composed of $r$ runs of letters from the
alphabet $\{0,1,\ldots,\sigma{-}1\}$ such that $2 \le \sigma \le r$, we describe
a data structure that, provided $r \le n / \log^{\omega(1)} n$, stores the
string in $r\log\frac{n\sigma}{r} + o(r\log\frac{n\sigma}{r})$ bits and supports
select and access queries in $O(\log\frac{\log(n/r)}{\log\log n})$ time and rank
queries in $O(\log\frac{\log(n\sigma/r)}{\log\log n})$ time. We
show that $r\log\frac{n(\sigma-1)}{r} - O(\log\frac{n}{r})$ bits are necessary for any such data
structure and, thus, our solution is succinct. We also describe a data structure
that uses $(1 + \epsilon)r\log\frac{n\sigma}{r} + O(r)$ bits, where
$\epsilon > 0$ is an arbitrary constant, with the same query times but
without the restriction $r \le n / \log^{\omega(1)} n$. By simple reductions to
the colored predecessor problem, we show that the query times are optimal in the
important case $r \ge 2^{\log^\delta n}$, for an arbitrary constant $\delta > 0$.
We implement our solution and compare it with the state of the art,
showing that the closest competitors consume 31--46\% more
space.
\end{abstract}

%\textbf{{Keywords: }}run length compression, succinct, rank select, optimal time

\section{Introduction}

Data structures supporting rank and select queries on sequences are fundamental to a
wide variety of topics in the theoretical and practical computer science, especially
as a component of more complex succinct and compressed data structures (we provide a
formal definition of rank and select queries below). Rank and select structures for
non-binary strings have been of interest since the advent of the FM-index~\cite{FerraginaManzini}
and the compressed suffix array~\cite{ggv}, and subsequent works on other indexes based on
the Burrows--Wheeler transform~\cite{BurrowsWheeler} (BWT) (e.g.,
see~\cite{MakinenNavarro}).

The simple run-length encoding of the BWT of a string allows to achieve, on
highly repetitive strings, compression ratios that are comparable to the
compression ratios achieved by the best reference-based schemes such as
LZ77~\cite{LZ77}. The crucial component required for the implementation of
compressed indexes based on BWT is the support of the rank (and,
sometimes, select) queries on the run-length encoded BWT. Many works have been
published on this and related topics (e.g.,
see~\cite{BelazzouguiNavarro,GolynskiMunroRao,MakinenNavarro,MakinenNavarroSirenValimaki}
and references therein), but none of them could achieve optimal (succinct)
run-length compressed space and optimal time simultaneously (recall that a data
structure is called \emph{succinct} if it occupies $Z + o(Z)$ space, where $Z$
is the information theoretic lower bound for its size). This line of research
became especially interesting after a recent paper by Gagie et
al.~\cite{GagieNavarroPrezza}, in which these authors proposed a method that
allows to avoid the $O(\frac{n}{\mathop{\mathrm{polylog}}(n)})$ bits of
redundancy that were required in all previous BWT-based indexes supporting the
full range of common search operations.

In this paper we describe a data structure that, given a string of length $n$
with $r$ runs of letters from the alphabet
$\{0,1,\ldots,\sigma{-}1\}$ such that $2 \le \sigma \le r \le n /
\log^{\omega(1)} n$, stores the string in $r\log\frac{n\sigma}{r} +
o(r\log\frac{n\sigma}{r})$ bits of
space (for brevity, $\log$ denotes the logarithm with base~2), and
supports select and access queries in $O(\log\frac{\log(n/r)}{\log\log n})$ time
and rank queries in $O(\log\frac{\log(n\sigma/r)}{\log\log n})$ time. Further, we
prove that $r\log\frac{n(\sigma - 1)}{r} - O(\log\frac{n}{r})$ bits are necessary for any such
encoding in the worst case and this implies that our data structure is
succinct. We also describe a version of this data structure that uses
$(1 + \epsilon)r\log\frac{n\sigma}{r} + O(r)$ bits, for arbitrary constant
$\epsilon > 0$, with the same query times but without the restriction
$r \le n / \log^{\omega(1)} n$. We then show, via
reductions to the so-called colored predecessor problem, that provided
$r \ge 2^{\log^\delta n}$ for an arbitrary constant $\delta > 0$, our
rank, select, and access times are optimal, even if $\sigma =
2$. We also describe a generic version of our solution with a parameter
controlling time-space trade-offs.

We have implemented our data structure and
%(with some simplifications necessary for practice)
experiments show that our closest competitors
consume 31\%--46\% more space; this small space usage, however, comes at the price of a noticeable slowdown
in query time on some inputs.
%~\cite{Barbay2014,Belazzougui2015,GolynskiMunroRao,MakinenNavarro}.

This paper is organized as follows. In the next section we first discuss some
auxiliary tools and then describe the main data structure.
The subsequent section is devoted to time and space optimality
considerations. The last section presents a practical implementation of these
ideas and experiments.
%We conclude with some remarks and open problems in the last section.

\myparagraph{Preliminaries.}
Let $s$ be a string of length $n$. The letters of $s$ are denoted
$s[0]$, $s[1]$, $\ldots$, $s[n{-}1]$ and $s[i..j]$ denotes $s[i]s[i{+}1]\cdots s[j]$.
Our notation for arrays is similar: e.g., $a[0..n{-}1]$ is an
array of length $n$. A \emph{run} in $s$ is
a substring $s[i..j]$ in which all letters are equal. For any $i,j$, the set
$\{k\in \mathbb{Z} \colon i \le k \le j\}$ is denoted
$[i..j]$.

For a string or an array $B$, the query $\access(i, B)$ merely returns $B[i]$,
the query $\rank_c(i, B)$ returns the number of letters $c$ in
$B[0..i]$, the query $\select_c(i, B)$ returns the position of the
$i$th letter $c$ in $B$ or returns $-1$ if either $i < 1$ or $B$ contains less
than $i$ letters $c$ (in particular, $\select_c(0, B) = -1$). We omit
the second parameter $B$ and write simply $\rank_c(i)$, $\select_c(i)$,
$\access(i)$ when $B$ is clear from the context.

\section{Data Structure}

We briefly describe the so-called \emph{Elias--Fano scheme}~\cite{Elias,Fano},
which is of fundamental importance for succinct data structures and which we
use in our construction below.

\myparagraph{Elias--Fano scheme.}
Consider a bit array of length $n$ that contains exactly $k$ ones. In the
Elias--Fano scheme we split the array into $k$ buckets of lengths
$\lceil\frac{n}{k}\rceil$ (the last bucket can be smaller), concatenate the
unary encodings of the numbers of ones in the buckets, thus obtaining a bit
array $B$ of length $2k$ (e.g., the bit array $001101$ encodes three buckets
containing, respectively, $2$, $0$, and $1$ ones), and finally, store
consecutively the relative positions of the ones inside the buckets in an array
$A[0..k{-}1]$. $A$ occupies $k\lceil\log\frac{n}{k}\rceil$ bits, and the whole
encoding takes $k\log\frac{n}{k} + O(k)$ bits in total. In addition, $B$ is
equipped with the following data structure, adding $o(k)$ bits.

\begin{lemma}[{see \cite{Jacobson}}]
Any bit array of length $n$ has an encoding that occupies $n + o(n)$ bits and
supports the queries $\rank_0$, $\rank_1$, $\select_0$, $\select_1$ in $O(1)$
time.
\label{RankSelectBitArray}
\end{lemma}

Using the data structure of Lemma~\ref{RankSelectBitArray}, one can compute the
bucket containing the $i$th one as $\select_{0}(i,B) - i + 1$. The relative
position of this one in the bucket is stored explicitly in $A[i]$. Therefore,
any $\select_1(i)$ query on the bit array can be answered in $O(1)$
time. We further enhance this basic scheme as follows.

\begin{lemma}[{see~\cite[Th. 14]{BelazzouguiNavarro}}]
Given a set $S \subset [0..u]$ of size $k$, there is a data structure that
occupies $O(k\log\frac{u}{k})$ bits and supports, for any given $x$, predecessor
queries $\max\{y \in S \colon y < x\}$ in $O(\log\log_w\frac{u}{k})$ time,
where $w$ is the size of machine word.
\label{PredecessorLemma}
\end{lemma}

\begin{lemma}
Let $\tau \ge 1$ be a ``sampling'' parameter.
Any bit array of length $n$ containing exactly $k$ ones has an encoding that
occupies $(1 + \frac{1}{\tau})k\log\frac{n}{k} + O(k)$ bits and supports
$\select_1$ queries in $O(1)$ time and $\rank_0$/$\rank_1$ queries in
$O(\log\tau + \log\frac{\log(n/k)}{\log\log n})$ time.
\label{FullEliasFano}
\end{lemma}
\begin{proof}
Since $\select_1$ was discussed above and $\rank_0(i) = i - \rank_1(i) + 1$, it
suffices to consider $\rank_1$. Let $b = \lfloor i / \lceil n/k\rceil\rfloor +
1$. Obviously, $i$ lies in the $b$th bucket and, hence, $d = \select_1(b - 1, B)
- b + 2$ ones from other buckets precede $i$. Thus, it remains to count the
number of ones before position $i$ that also lie in the $b$th bucket.

It is easy to see that there are $\ell = \select_1(b, B) - \select_1(b - 1, B) -
1$ ones in the $b$th bucket and their relative positions are stored in the
subarray $A[d .. d + \ell - 1]$. Denote $\rho = c\tau$, where $c$ is a positive
constant determined below. If $\ell \le \rho$, we use binary search
to count in $O(\log\rho) = O(\log\tau)$ time the number of subarray elements that precede
the relative position of $i$ inside the $b$th bucket, i.e., precede $i \bmod
\lceil n/k\rceil$. For the case $\ell > \rho$, we sample every
$\lceil\rho\rceil$th element of the subarray and put them in the data structure
from Lemma~\ref{PredecessorLemma} occupying
$O(\frac{k_b}{\rho}\lceil\log\frac{n}{k}\rceil)$
bits, where $k_b$ is the number of ones in the $b$th bucket, which allows us
to count the number of sampled predecessors of $i \bmod \lceil n/k\rceil$ in
$O(\log\frac{\log(n/k)}{\log\log n})$ time; the $\lceil\rho\rceil{-}1$
non-sampled elements following the found sampled predecessor are processed again
by binary search. We store these data structures consecutively and
locate the required one using an additional bit array of length
$k$ that marks the sampled elements of $A$; the details are omitted as they are
straightforward. The overall space can be estimated as
$k\log\frac{n}{k} + O(k + \frac{k}{\rho}\log\frac{n}{k})$, which is
$(1 + \frac{1}{\tau})k\log\frac{n}{k} + O(k)$ for an appropriately chosen
constant $c$ in $\rho = c\tau$.
\end{proof}

For example, when $\tau = \log n$, Lemma~\ref{FullEliasFano} gives us a data
structure that occupies $k\log\frac{n}{k} + O(k)$ bits and answers rank queries
in $O(\log\log n)$ time.

\myparagraph{The main data structure.}

Let us consider a string $s$ of length $n$ that can be represented as a
concatenation of $r$ runs of letters from the alphabet $[0..\sigma{-}1]$ such
that $2 \le \sigma \le r$. Denote by $n_0, n_1, \ldots, n_{\sigma-1}$ the
number of runs of the letters $0, 1, \ldots, \sigma{-}1$, respectively; note
that $n_0 + n_1 + \cdots + n_{\sigma-1} = r$. For $c \in [0..\sigma{-}1]$ and $i
\in [1..n_c]$, let $\ell_{c,i}$ denote the length of the $i$th run of the letter
$c$. We encode $s$ in the following components (see a clarifying example in
Fig.~\ref{fig:example}):
\begin{itemize}
\item a bit array $R[0..n{-}1]$ such that $R[i] = 1$ iff $s[i] \ne s[i{+}1]$ or
  $i = n - 1$;
\item a string $H[0..r{-}1]$ such that $H[i] = s[\select_1(i + 1, R)]$;
\item a bit array $C[0..r{+}\sigma{-}1]$ that is the concatenation of the unary
  encodings for the number of runs of each letter;
\item an integer array $S[0..r{-}1]$ that stores the following numbers (in this
  order):\\ $ \ell_{0,1}, \ell_{0,2}, \ldots, \ell_{0,n_0},\ \ \ell_{1,1},
  \ell_{1,2}, \ldots, \ell_{1,n_1},~~~\ldots~~~,\ell_{\sigma-1,1},
  \ell_{\sigma-1,2}, \ldots, \ell_{\sigma-1,n_{\sigma-1}}.  $
\end{itemize}

\begin{wrapfigure}{r}{0.45\textwidth}
\vspace{-5.7ex}
$$
\begin{array}{rl}
s = & aaaabbbadddddaaaaaddbaaaa, \\
R = & 0001001100001000010110001, \\
H = & abadadba, \\
C = & 000010011001, \\
S = & 4,1,5,4,3,1,5,2.
\end{array}
$$
\vspace{-3ex}
\caption{\footnotesize Here $\sigma = 4$ and, for the readability, $a, b, c, d$
  denote the letters $0, 1, 2, 3$.}
\vspace{-2ex}\label{fig:example}
\end{wrapfigure}

Thus, $R$ marks the last letter of every run in $s$, $H$ stores these letters in
the corresponding order, $C$ encodes the number of runs of each letter, and $S$
stores the run lengths grouped by letters.

Let us choose a positive ``sampling'' parameter $\rho \ge 1$ that will regulate
time-space trade-offs for the data structure. Our goal is to support the queries
$\select_c$ and $\access$ in $O(\rho + \log\frac{\log(n/r)}{\log\log n})$
time and the query $\rank_c$ in $O(\rho + \log\frac{\log(n\sigma/r)}{\log\log n})$
time (see optimality considerations below).

We encode $R$ in $(1 + \frac{1}{2^\rho})r\log\frac{n}{r} + O(r)$ bits as in
Lemma~\ref{FullEliasFano} with $\tau = 2^\rho$. The array $C$ is encoded in $O(r)$ bits
as in Lemma~\ref{RankSelectBitArray}. The string $H[0..r{-}1]$
is stored in the following data structure of
Belazzougui and Navarro~\cite{BelazzouguiNavarro} (slightly reformulated).

\begin{lemma}[{see~\cite[Th. 6]{BelazzouguiNavarro}}]
Let $\rho \ge 1$ be a ``sampling'' parameter.
Any string of length $r$ over the alphabet $[0..\sigma{-}1]$ such that $2 \le
\sigma \le r$ has an encoding that occupies $(1 + \frac{1}{\rho})r\log\sigma + o(r\log\sigma) + O(r)$ bits
and supports $\access$ in $O(\rho)$ time, $\select_c$ in $O(1)$ time, and
$\rank_c$ in $O(\log\frac{\log\sigma}{\log\log n})$ time.
\label{StringRankSelect}
\end{lemma}
\begin{proof}
The result follows from the proof of \cite[Theorem~6]{BelazzouguiNavarro}
if we put $f(n,\sigma) = \rho$.
\end{proof}

The structures $R$ and $H$ are already sufficient to implement $\access$ queries
in $O(\rho + \log\frac{\log(n/r)}{\log\log n})$ time since
$s[i] = H[\rank_1(i{-}1, R)]$. For $\rank_c$ and
$\select_c$, we need $C$ and $S$. It turns out that we do not have to store $S$
explicitly since, as it is shown below, $S[i]$ can be computed in $O(1)$ time,
for any $i \in [0..r{-}1]$, using $R$, $H$, and $C$. However, besides access,
our data structure requires to answer on $S$ the queries $\pred(x, S)$ that
return the maximal $i \in [0..r{-}1]$ such that $S[0] + S[1] + \cdots + S[i] < x$.

\begin{wraptable}{r}{70mm}
\vspace{-2ex}
\caption{\footnotesize Component sizes.}\label{tab:sizes}
\vspace{-1ex}
\begin{tabular}{cc}\toprule
~ & size in bits \\\midrule
$R$ & $(1 + \frac{1}{2^\rho})r\log\frac{n}{r} + O(r)$ \\
$H$ & $(1{+}\frac{1}{\rho})r\log\sigma + o(r\log\sigma) + O(r)$ \\
$C$ & $O(r)$ \\
$S$ & $O(\frac{r}{\rho}\log\frac{n\rho}{r})$\\\bottomrule
\end{tabular}
\vspace{-3ex}
\end{wraptable}

In order to perform $\pred$ in small space and $O(\rho +
\log\frac{\log(n/r)}{\log\log n})$ time, we sample the numbers $S[0] + S[1] +
\cdots + S[j\lceil\rho\rceil]$, for $j \in [0..\frac{r-1}{\lceil\rho\rceil}]$,
and store them in the data structure from Lemma~\ref{PredecessorLemma},
thus consuming $O(\frac{r}{\rho}\log\frac{n\rho}{r})$ bits.
To perform $\pred(x, S)$, we first find in
$O(\log\frac{\log(n\rho/r)}{\log\log n}) \le O(\rho +
\log\frac{\log(n/r)}{\log\log n})$ time the maximal $j$ such that $S[0]
+ S[1] + \cdots + S[j\lceil\rho\rceil] < x$ and, then, compute
$S[j\lceil\rho\rceil{+}1], S[j\lceil\rho\rceil{+}2], \ldots,
S[(j{+}1)\lceil\rho\rceil]$ in $O(\rho)$ time, hence finding the answer in an
obvious way.

The sizes of the described data structures are summarized in
Table~\ref{tab:sizes}. Before discussing the implementation of $\rank_c$ and
$\select_c$, let us explain how one can compute $S[i]$ in $O(1)$ time using
$R$, $H$, and $C$.

It follows from the definition of $C$ that $S[i]$ stores the length of a run of
the letter $c = \rank_1(\select_0(i + 1, C), C)$. Further, it is straightforward
that there are exactly $j = \select_1(c, C) - c + 1$ runs of letters $0, 1,
\ldots, c{-}1$ and therefore, by definition, $S[j]$ stores the length of the
leftmost run of $c$. Hence, $S[i]$ stores the length of the $k$th run of $c$,
where $k = i - j + 1$; we compute $k$ in $O(1)$ time. Then, we find $k' =
\select_c(k, H) + 1$ in $O(1)$ time (see Lemma~\ref{StringRankSelect}). Clearly, the
$k$th run of $c$ is the $k'$th run (of all runs) and its length, which is equal
to $S[i]$, can be calculated as $\select_1(k', R) - \select_1(k' - 1, R)$ in
$O(1)$ time.

Consider the $\select_c(i, s)$ query. Put $j = \select_1(c, C) - c + 1$. As
above, $S[j]$ is the length of the leftmost run of the letter $c$ (if any). Let
us find the maximal $k$ such that $S[j] + S[j{+}1] + \cdots + S[j{+}k{-}1] < i$;
obviously, the $i$th occurrence of $c$ (if any) must lie in the $(k{+}1)$st run
of $c$. As $k = \pred(S[0] + \cdots + S[j{-}1] + i, S) - j + 1$, it suffices to
show how to compute $t = S[0] + \cdots + S[j{-}1]$. We calculate $t$ by
summing $S[j'\lceil\rho\rceil{+}1] + S[j'\lceil\rho\rceil{+}2] + \cdots + S[j{-}1]$,
where $j' = \lfloor (j - 1)/\lceil\rho\rceil\rfloor$, with the $(j'{+}1)$st number
sampled from $S$, which, by definition, is equal to $S[0] + S[1] + \cdots +
S[j'\lceil\rho\rceil]$, all in $O(\rho + \log\frac{\log(n/r)}{\log\log n})$ time.
Further, the $(k{+}1)$st run of $c$ exists iff $c =
\rank_1(\select_0(j + k + 1, C), C)$;
we check this condition and return $-1$ if
the run does not exist. Otherwise, we calculate the sum $t' = S[0] + S[1] +
\cdots + S[j{+}k{-}1]$ in $O(\rho + \log\frac{\log(n/r)}{\log\log n})$ time in the
same way as we computed $t$; then, the $i$th occurrence of $c$ in the string $s$ must
be the $p$th letter, where $p = i - (t' - t)$, of the $(k{+}1)$st run of $c$; thus,
we obtain $\select_c(i, s) = \select_1(\select_c(k + 1, H), R) + p$.

Consider the $\rank_c(i, s)$ query. Put $j = \select_1(c, C) - c + 1$. Again,
$S[j]$ is the length of the leftmost run of the letter $c$ (if any). In
$O(\rho + \log\frac{\log(n/r)}{\log\log n})$ time we compute $m = \rank_1(i - 1, R)$, which is the number of
runs preceding the position $i$ (excluding the run containing~$i$). Then, $k =
\rank_c(m - 1, H)$ of them are runs of $c$; $k$ is computed in
$O(\log\frac{\log\sigma}{\log\log n})$ time by Lemma~\ref{StringRankSelect}.
Thus, the total length of the runs of $c$ preceding the position $i$ can be calculated
as $x = S[j] + S[j{+}1] + \cdots + S[j{+}k{-}1]$ in $O(\rho + \log\frac{\log(n/r)}{\log\log n})$ time
(as in $\select_c$ above). It remains to
check whether the position $i$ itself is inside a run of $c$: it is so iff $H[m]
= c$. Accordingly, we return $x + (i - \select_1(m, R))$ if $H[m] = c$ (here $i - \select_1(m, R)$
is the position of $i$ in the run), and $x$ otherwise.

\begin{lemma}
Let $\tau \ge 1$ be a ``sampling'' parameter.
Any string of length $n$ with $r$ runs over the alphabet $[0..\sigma{-}1]$ such
that $2 \le \sigma \le r$ has an encoding that occupies $(1 + \frac{1}{\tau})r\log\frac{n\sigma}{r} + o(r\log\sigma)
+ O(r)$ bits and supports the queries $\select_c$ and $\access$ in $O(\tau +
\log\frac{\log(n/r)}{\log\log n})$ time and $\rank_c$ in $O(\tau + \log\frac{\log(n\sigma/r)}{\log\log n})$ time.
\label{MainLemma}
\end{lemma}
\begin{proof}
The space required for $S$ is
$O(\frac{r}{\rho}\log\frac{n\rho}{r}) = O(\frac{r}{\rho}\log\frac{n}{r} + r\frac{\log\rho}{\rho})
\le O(\frac{r}{\rho}\log\frac{n}{r}) + O(r)$. Summing up the space bounds from Table~\ref{tab:sizes}, we obtain
$(1 + O(\frac{1}{\rho}))r\log\frac{n\sigma}{r} + o(r\log\sigma)
+ O(r)$ bits. Further, $\select_c$ and $\access$ run in $O(\rho + \log\frac{\log(n/r)}{\log\log n})$
time; $\rank_c$ takes $O(\rho + \log\frac{\log(n/r)}{\log\log n} +
\log\frac{\log\sigma}{\log\log n})$ time, which can be simplified as $O(\rho +
\log\frac{\log(n\sigma/r)}{\log\log n})$. Putting $\rho = c\tau$ for an appropriate
constant $c$, we obtain the result.
\end{proof}

%Putting $\tau = \log\log n$ in Lemma~\ref{MainLemma}, we can obtain
%a data structure that occupies $r\log\frac{n\sigma}{r} + o(r\log\frac{n\sigma}{r})$
%bits, provided $r = o(n)$ or $\sigma = \omega(1)$, and supports all queries in
%$O(\log\log n)$ time. It is shown below (Theorems~\ref{SpaceLowerBound}
%and~\ref{LowerBoundThm}) that such structure is space optimal
%but is time optimal only for $r \approx n^{\epsilon}$, where $\epsilon \in (0,1)$ is a
%constant. Lemma~\ref{MainLemma} implies many other trade-offs
%but we especially distinguish the following one.

\begin{theorem}
Any string of length $n$ with $r$ runs over the alphabet $[0..\sigma{-}1]$ such
that $2 \le \sigma \le r \le n / \log^{\omega(1)} n$ has an encoding that occupies
$r\log\frac{n\sigma}{r} + o(r\log\frac{n\sigma}{r})$ bits and supports
$\select_c$ and $\access$ in $O(\log\frac{\log(n/r)}{\log\log n})$
time and $\rank_c$ in $O(\log\frac{\log(n\sigma/r)}{\log\log n})$ time.
\label{MainTheorem}
\end{theorem}
\begin{proof}
The result follows from Lemma~\ref{MainLemma} with $\tau = \log\frac{\log(n/r)}{\log\log n}$
since, for $r \le n / \log^{\omega(1)} n$, we have $\tau = \omega(1)$ and $r = o(r\log\frac{n}{r})$.
\end{proof}

\begin{theorem}
Any string of length $n$ with $r$ runs over the alphabet $[0..\sigma{-}1]$ such
that $2 \le \sigma \le r$ has an encoding that occupies
$(1 + \epsilon)r\log\frac{n\sigma}{r} + O(r)$ bits, where $\epsilon$ is
an arbitrary positive constant, and supports $\select_c$ and $\access$ queries in
$O(\log\frac{\log(n/r)}{\log\log n})$ time and $\rank_c$ queries in
$O(\log\frac{\log(n\sigma/r)}{\log\log n})$ time.
\label{MainTheorem2}
\end{theorem}
\begin{proof}
Since $o(r\log\sigma) \le \frac{1}{2\epsilon}r\log\frac{n\sigma}{r}$ for large enough $n$,
the result follows from Lemma~\ref{MainLemma} with $\tau = \frac{1}{2\epsilon}$; the
big-O notation hides the additive constant $\frac{1}{\epsilon}$ in the time bounds.
\end{proof}

Lemma~\ref{MainLemma} implies many other trade-offs that we do not discuss separately.

\section{Optimality}

Clearly there is a one-to-one correspondence between the set $T$
of all strings of length $n$ with $r$ runs over the alphabet $[0..\sigma{-}1]$
and the pairs $(R, H)$ such that $R[0..n{-}2]$ is a bit array with $r - 1$ ones and
$H[0..r{-}1]$ is a string such that $H[i] \ne H[i{+}1]$, for $i \in
[0..r{-}2]$. Hence, the size of $T$ is ${n - 1 \choose r - 1}\sigma(\sigma -
1)^{r-1}$. Since $\log(x - \frac{1}{r}) \ge \log x - \frac{2}{r}$ for any $x \ge 1$
and $r \ge 2$, we obtain $\log{n - 1 \choose r - 1} \ge (r - 1)\log\frac{n - 1}{r - 1} \ge
(r - 1)\log(\frac{n}{r} - \frac{1}{r}) \ge (r - 1)(\log\frac{n}{r} - \frac{2}{r}) =
r\log\frac{n}{r} - O(\log\frac{n}{r})$ and thus
$\log|T| \ge r\log\frac{n}{r} - O(\log\frac{n}r) + \log((\sigma - 1)^r) = r\log\frac{n(\sigma - 1)}{r} - O(\log\frac{n}r)$,
which implies the following lower bound and that therefore the data structure of
Theorem~\ref{MainTheorem} is succinct.

\begin{theorem}
Any encoding of a string of length $n$ with $r$ runs over an alphabet of size
$\sigma$ requires at least $r\log\frac{n(\sigma - 1)}{r} - O(\log\frac{n}r)$ bits in the worst
case.\label{SpaceLowerBound}
\end{theorem}

Let us investigate the optimality of the query times provided in
Theorems~\ref{MainTheorem} and~\ref{MainTheorem2}. For this, we use reductions
to the well-known \emph{colored predecessor data structure}, in which one is given
a set of $r$ integers from the universe $[0..n]$ each of which is colored either
in red or blue, and the query asks to find, for a given integer $x$, the color of
the maximal $y$ from this set such that $y \le x$. Our reductions resemble those
from~\cite[section~7]{RamanRamanRao} but we, nevertheless, present them for
completeness.

\begin{lemma}
Suppose that, for any binary string of length $n$ with $r$ runs, there is an
encoding that occupies $O(r\log^{O(1)} n)$ bits and supports $\rank_c$, $\select_c$,
and $\access$ queries in, respectively, $t_r$, $t_s$, and $t_a$ time. Then,
there is a colored predecessor data structure that stores $r$ integers (colored
in red or blue) from the universe $[0..n]$ in $O(r\log^{O(1)} n)$ bits of space and
supports the colored predecessor queries in $O(\min\{t_r, t_s, t_a\})$ time.
\label{ReductionToPred}
\end{lemma}
\begin{proof}
Let $x_1, x_2, \ldots, x_r$ be the integers stored in our colored predecessor
data structure. Suppose $t_a \le \min\{t_r, t_s\}$. We create a bit string
$s[0..n]$ such that $s[i] = 1$ iff the predecessor of $i$ is colored in
red. Then, the colored predecessor queries can be answered in $O(t_a)$ time
using our $O(r\log^{O(1)} n)$-bit encoding and $\access$ queries.

Suppose $t_r \le \min\{t_s, t_a\}$. Then, we store the colors of $x_1, x_2,
\ldots, x_r$ in an array $c[0..r{-}1]$ and create a bit string $s[0..n]$ such
that $s[i] = 1$ iff $i = x_j$ for some $j \in [1..r]$. Thus, the colored
predecessor query can be answered as $c[\rank_1(x, s){-}1]$ in $O(t_r)$ time.

Finally, suppose $t_s \le \min\{t_r, t_a\}$. We create again the array $c$ and
create a bit string $s[0..n{+}r]$ such that $s[i] = 1$ iff $i = x_j + j - 1$ for
some $j \in [1..r]$. Then, the colored predecessor query can be answered as
$c[\select_0(x, s) - x]$ in $O(t_s)$ time.
\end{proof}

Assuming
$r \ge 2^{\log^\delta n}$ for a constant $\delta > 0$ and
putting $n' := r$, $S := r\log^{O(1)} n$, $w := \Theta(\log n)$, $\ell := \log n$
in the formula of P{\u{a}}tra{\c{s}}cu and Thorup~\cite{PatrascuThorup}
(we denote their $n$ by $n'$ to distinguish it from our $n$), we obtain the lower bound
$\Omega(\min\{\frac{\log r}{\log\log n}, \log\frac{\log(n/r)}{\log\log n},$ $\log\log n, \log\log n\})
= \Omega(\log\frac{\log(n/r)}{\log\log n})$, which holds, by Lemma~\ref{ReductionToPred},
for $\rank_c$, $\select_c$, and $\access$ in any data structure occupying
$O(r\log^{O(1)} n)$ bits.
Combining this with the lower bound $\Omega(\log\frac{\log\sigma}{\log\log n})$
from~\cite{BelazzouguiNavarro} for rank, we deduce the following theorem.

\begin{theorem}
Any data structure that stores a string of length $n$ with $r$ runs in $O(r\log^{O(1)} n)$ bits
requires $\Omega(\log\frac{\log(n\sigma/r)}{\log\log n})$ time for $\rank_c$ and
$\Omega(\log\frac{\log(n/r)}{\log\log n})$ time for $\select_c$ and $\access$ in the
worst case, provided $r \ge 2^{\log^\delta n}$ for a constant $\delta \in (0,1)$.%
\label{LowerBoundThm}%
\end{theorem}

Theorem~\ref{LowerBoundThm} implies that the data structures of
Theorems~\ref{MainTheorem} and~\ref{MainTheorem2} are time optimal whenever
$r \ge 2^{\log^\delta n}$, for an arbitrary positive constant $\delta \in (0,1)$.

\section{Experimental Evaluation}

We implemented our data structure and measured its practical performance relative
to other rank and select data structures available in the Succinct Data Structures
Library (SDSL)~\cite{gbmp2014sea}, including the data structures of:
%To demonstrate the practical impact of our solution, we provide an
%implementation of it and compare it against the state of the
%art. For comparison, we use several implementations available in the library
%{\sc SDSL}~\cite{gbmp2014sea}: the data structures of
Golynski et al.~({\tt gmr})~\cite{GolynskiMunroRao},
Barbay et al.~({\tt ap})~\cite{Barbay2014}, and M\"{a}kinen
and Navarro ({\tt
  rlmn})~\cite{Makinen:2005:SSA:1195881.1195885}. The implementation
of {\tt gmr} uses $n\log\sigma + o(n\log\sigma)$ bits and supports access, rank,
and select in, resp., $O(\log\log\sigma)$,
$O(\log\log\sigma)$, and $O(1)$ time. The implementation of {\tt ap} uses
$nH_0+o(n)(H_0+1)$ bits and supports access, rank, and select in
$O(\log\log\sigma)$ worst-case time or $O(\log H_0)$ average time. The
implementation of {\tt rlmn} uses $2r(2+\log(n/r))+\sigma\log n+u$ bits,
where $u$ is the space of an underlying rank/select structure over a sequence
of length $r$, and supports access, rank, and select in, resp.,
$O(\log(n/r)+u_a)$, $O(\log(n/r)+u_a+u_r)$, and $O(\log(n/r)+u_s)$ time,
where $u_a$, $u_r$, and $u_s$ correspond to the access, rank, and select
times of the underlying structure; we use {\tt ap} as the underlying
structure as it showed the best time-space trade-offs. Also, we tested the
data structure of Belazzougui et al.~({\tt rle})~\cite{Belazzougui2017}, which uses
$(1+\gamma)r\log\frac{n\sigma}{r} + O(r)$ bits and supports access, rank, and
select in $O(\frac{1}{\gamma}\log\frac{n}{r})$ time, for any $\gamma \in (0,1)$; {\tt rle}
is similar to our solution\footnote{We thank the reviewers for informing us about this data structure.}
but it was implemented only for small alphabets (hence we could not apply it
for all datasets). Additionally, we implement a simple construction
of~\cite{Belazzougui2015}~({\tt bcgpr}): it uses $O(r\log n)$
bits and supports rank and select in $O(\log\log n)$ time (access was not
considered originally). The construction contains
a pair of predecessor data structures for each letter $c \in [0..\sigma{-}1]$: the
first predecessor structure stores the starting indexes of runs of $c$ and the
second one stores the number of letters $c$ before the starting index of each run
of $c$. Our implementation of {\tt bcgpr} uses binary searches instead of the predecessor
structures and additionally we store the bit array $R$ and the string $H$ of our solution, but
without rank/select support for $H$, in order to support access; thus, access
takes $O(\log\frac{\log(n/r)}{\log\log n})$ time and rank/select queries take $O(\log r)$ time.

\myparagraph{Implementation.} %Our implementation uses components from the SDSL.
Our solution implements $R$ using Elias--Fano sparse bit array
from the {\sc SDSL}; for $H$, we used the {\sc SDSL} implementations
of {\tt gmr} and {\tt ap} for large alphabets, and {\tt huff} for small alphabets;
$C$ was encoded as a plain bit
array supporting rank and select in $O(1)$ time; finally, $S$ was
stored as an integer array with samples to support $\pred$ queries via binary search
over $S$. The queries access, rank, and select were implemented
verbatim. In the experiments, we call our solution
{\tt fkkp}\footnote{The implementation is available at \url{https://github.com/jfuentess/sdsl-lite}}.

\myparagraph{Experimental setup.}  The experiments were carried out on an
Intel\textregistered{} Core\textregistered{} i7-7700 machine with 4 physical
cores, clocking at 3.6GHz each, with one 32KB L1 instruction cache per core, one
32KB L1 data cache per core, one 256KB L2 cache per core, and
one 8MB shared L3 cache.
%between the 8 cores.
The code of all the structures was compiled with {\tt g++} and optimization
level -O3.
The data structures were
compared in terms of query times using the high-resolution C++ function {\tt
  high\_resolution\_clock} in the {\tt <chrono>} library. The space consumption
was measured by the serialization of data structures to their binary format.
Experiments constructed each data structure on several
datasets with varying $n$ and $\sigma$.
%As our data structure, {\tt rle} uses a sampling step to provide a time-space trade-off.
%For the size of the sampling step, we tested several powers of two but we report only
We tested several power of two sampling steps for our structure and {\tt rle}, but we report only
% the values 4, 8, 16, 32, 64, 128, and 256. We tested {\tt gmr}, {\tt ap},
% and {\tt rlmn} as the underlying structure in the implementation of
% the string $H$. However, we only report results for {\tt gmr} and {\tt
%   ap}, and only for the sampling steps
4, 16, and 32 since they exhibited the best time-space trade-offs. To
differentiate each configuration, we use the
name {\tt fkkp\_x\_y} and {\tt rle\_x} to denote a sampling step of {\tt x} and
underlying structure {\tt y}.

\begin{wraptable}{r}{63mm}
%\begin{table}[th]
\vspace*{-3ex}
\caption{\footnotesize Datasets used in experiments.}
\vspace{-3ex}
\begin{center}
  \footnotesize
  \begin{tabular}{@{\hspace{-0.1ex}}l@{\hspace{1ex}}r@{\hspace{2ex}}r@{\hspace{2ex}}r@{\hspace{-0.1ex}}}
     \toprule 	
     dataset & $n$ & $\sigma$ & runs \\
     \midrule 	
      {\tt wl\_1B} & 46,968,182 & 90 & 573,487 \\
      {\tt wl\_2B} & 46,968,182 & 2,528 & 875,406 \\
      {\tt wiki} & 140,990,835 & 174,796 & 2,586,752 \\
      {\tt kr\_1B} & 257,961,617 & 161 & 2,791,368 \\
      {\tt kr\_2B} & 257,961,617 & 7,124 & 4,194,799 \\
     \bottomrule 	
   \end{tabular}
\end{center}
\vspace*{-3.5ex}
\label{tbl:datasets}
%\end{table}
\end{wraptable}

The datasets are shown in Table~\ref{tbl:datasets}.\footnote{Available online at \url{https://users.dcc.uchile.cl/~jfuentess/datasets/sequences.php}}
All our datasets are the BWT of highly repetitive sequences.
The datasets {\tt wl\_1B} and {\tt wl\_2B} were generated by taking
the previous and the two previous symbols during the BWT computation of the
sequence {\em world-leaders} from the {\em Pizza\&Chili} repetitive
corpus\footnote{\url{http://pizzachili.dcc.uchile.cl/repcorpus}. Last access:
  Nov.~2, 2017.}. Similarly, the datasets {\tt kr\_1B} and {\tt
  kr\_2B} were generated from the repetitive sequence {\em kernel}
from Pizza\&Chili. The dataset {\tt wiki} was generated from the edit history
of some Wikipedia pages in which words were used as letters\footnote{\url{https://dumps.wikimedia.org/enwiki/20171001/enwiki-20171001-pages-meta-history10.xml-p3037476p3046511.7z}. Last
  access: Nov. 2, 2017.}.\\

%We distinguish two kinds of
%datasets: For the first kind we generated the BWT of a highly
%repetitive sequence, taking the previous and the two previous symbols. In
%particular, the datasets {\tt wl\_1B} and {\tt wl\_2B} were generated by taking
%the previous and the two previous symbols during the BWT computation of the
%sequence {\em world-leaders} from the {\em Pizza\&Chili} repetitive
%corpus\footnote{\url{http://pizzachili.dcc.uchile.cl/repcorpus}. Last access:
%  Nov.~2, 2017.}. Similarly, the datasets {\tt kr\_1B} and {\tt
%  kr\_2B} were generated from the repetitive sequence {\em kernel}
%from Pizza\&Chili. The second kind of datasets corresponds to the BWT of a text
%file using words as symbols. The dataset {\tt wiki} belongs to this kind of
%dataset and was generated from the edit history of some Wikipedia
%pages\footnote{\url{https://dumps.wikimedia.org/enwiki/20171001/enwiki-20171001-pages-meta-history10.xml-p3037476p3046511.7z}. Last
%  access: Nov. 2, 2017.}.

\myparagraph{Results.}
Table~\ref{tbl:space} shows the size of each data structure for all the
datasets. In our experiments, the structure {\tt fkkp\_ap\_32} provides the best space
consumption, except for the dataset {\tt wl\_1B}, where {\tt rle\_32} has the
best consumption. For the underlying structures {\tt gmr}, {\tt ap}, and {\tt huff}, our
structure reduces its size by increasing the size of the sampling step. For small
alphabets, we are comparable with {\tt rle}. For large alphabets, the
closest competitor, {\tt rlmn}, uses from 31\% to 46\% more space than {\tt
  fkkp\_ap\_32}.

\begin{table}[t]
\caption{\footnotesize Space usage of the data structures in megabytes. Best results are underlined.}
\vspace{-3.5ex}
\begin{center}
  \scriptsize
   \begin{tabular}{l@{\hspace{2ex}}r@{\hspace{2ex}}r@{\hspace{2ex}}r@{\hspace{2ex}}r@{\hspace{2ex}}r@{\hspace{2ex}}r@{\hspace{2ex}}r@{\hspace{2ex}}r@{\hspace{2ex}}r@{\hspace{2ex}}r@{\hspace{2ex}}r@{\hspace{2ex}}r@{\hspace{2ex}}r@{\hspace{2ex}}r@{\hspace{2ex}}r@{\hspace{2ex}}r@{\hspace{2ex}}}
     \toprule
     & \multicolumn{3}{c}{{\tt fkkp\_gmr}} & \multicolumn{3}{c}{{\tt fkkp\_ap}} &
     \multicolumn{3}{c}{{\tt fkkp\_huff}} & \multirow{2}{*}{{\tt gmr}} &
     \multirow{2}{*}{{\tt ap}} & \multirow{2}{*}{{\tt rlmn}} & \multirow{2}{*}{{\tt bcgpr}} & \multicolumn{3}{c}{{\tt rle}} \\
     \cmidrule(r){2-4} \cmidrule(r){5-7} \cmidrule(r){8-10}\cmidrule(r){15-17}
     & {\tt 4} & {\tt 16} & {\tt 32} & {\tt 4} & {\tt 16} & {\tt 32} & {\tt 4} &
              {\tt 16} & {\tt 32} & & & & & {\tt 4} & {\tt 16} & {\tt 32} \\
     \midrule
     {\tt wl\_1B} & 2.21 & 1.80 & 1.73 & 1.82 & 1.41 & 1.34 & 1.79 & 1.38 & 1.31 & 78.71 & 26.99 & 1.93 & 9.96 & 1.40 & 1.26 & \underline{1.22} \\
     {\tt wl\_2B} & 4.15 & 3.53 & 3.42 & 3.24 & 2.61 & \underline{2.51} &  &  &  & 125.57 & 42.60 & 3.67 & 15.97 &  &  &  \\
     {\tt wiki} & 13.15 & 11.52 & 11.27 & 11.59 & 9.96 & \underline{9.71} &  &  &  & 402.42 & 274.61 & 12.75 & 53.56 &  &  &  \\
     {\tt kr\_1B} & 10.55 & 8.55 & 8.22 & 8.47 & 6.47 & \underline{6.14} & 8.53 & 6.54 & 6.20 & 440.30 & 230.83 & 8.84 & 48.29 & 7.18 & 6.53 & 6.40 \\
     {\tt kr\_2B} & 19.79 & 16.79 & 16.29 & 15.43 & 12.44 & \underline{11.94} &  &  &  & 698.00 & 410.50 & 16.21 & 76.43 &  &  &  \\
     \bottomrule 	
   \end{tabular}
\end{center}
\label{tbl:space}
\vspace*{-3.5ex}
\end{table}

\begin{table}[t]
  \setlength\tabcolsep{2.5pt}
  \scriptsize
  \caption{\footnotesize Running times for the access, rank, and select queries in $\mu$s. Best times are underlined.}
  \vspace{-3.5ex}
  \begin{center}
   \begin{tabular}{lrrrrr|rrrrr|rrrrr}
     \toprule 	
      & \multicolumn{5}{c}{access} & \multicolumn{5}{c}{rank} & \multicolumn{5}{c}{select}\\
     \cmidrule(lr){2-6}\cmidrule(lr){7-11}\cmidrule(lr){12-16}
     & {\tt wl\_1B} & {\tt wl\_2B} & {\tt wiki} & {\tt kr\_1B} & {\tt kr\_2B} &
              {\tt wl\_1B} & {\tt wl\_2B} & {\tt wiki} & {\tt kr\_1B} & {\tt
                kr\_2B} & {\tt wl\_1B} & {\tt wl\_2B} & {\tt wiki} & {\tt kr\_1B} & {\tt kr\_2B}\\
     \midrule
     {\tt fkkp\_gmr\_4} & .18 & .49 & .66 & .26 & 1.03 & .86 & 2.04 & 2.02 & .99 & 2.64 & .92 & 1.32 & 1.30 & .99 & 1.45 \\
     {\tt fkkp\_gmr\_16} & .18 & .49 & .65 & .26 & 1.03 &2.16 & 4.62 & 3.71 & 2.12 & 4.85 & 2.18 & 3.57 & 3.02 & 2.08 & 3.48 \\
     {\tt fkkp\_gmr\_32} & .18 & .49  & .64 & .26 & 1.03 & 3.42 & 6.61 & 5.81 & 3.50 & 7.58 & 3.89 & 6.49 & 5.20 & 3.52 & 6.16 \\
%     {\tt fkkp\_gmr\_256} & .18 & .49 & .64 & .26 & 1.03 & 25.43 & 67.70 & 31.41 & 21.67 & 42.17 & 26.64 & 46.98 & 32.16 & 22.79 & 42.54 \\
     {\tt fkkp\_ap\_4} & .26 & .42 & .64 & .32 & .62 &1.15 & 1.54 & 2.82 & 1.39 & 2.73 & 1.23 & 1.66 & 3.03 & 1.57 & 2.78 \\
     {\tt fkkp\_ap\_16} & .26 & .42 & .64 & .32 & .62 & 3.62 & 3.93 & 6.88 & 3.82 & 6.69 & 3.06 & 4.22 & 6.84 & 3.76 & 6.31 \\
     {\tt fkkp\_ap\_32} & .26 & .42 & .63 & .32 & .62 & 5.64 & 6.61 & 11.49 & 6.79 & 11.38 & 5.51 & 7.50 & 11.45 & 6.63 & 10.98 \\
%     {\tt fkkp\_ap\_256} & .26 & .42 & .64 & .32 & .62 & 37.89 & 47.21 & 57.72 & 46.40 & 66.90 & 38.51 & 51.53 & 58.42 & 46.19 & 73.26 \\
     {\tt fkkp\_huff\_4} & .11 &  &  & .17 &  & .68 &  &  & .84 &  & .88 &  &  & 1.11 &  \\
     {\tt fkkp\_huff\_16} & .11 &  &  & .17 &  & 2.14 &  &  & 2.42 &  & 2.17 &  &  & 2.54 &  \\
     {\tt fkkp\_huff\_32} & .11 &  &  & .17 &  & 3.48 &  &  & 4.37 &  & 3.93 &  &  & 4.47 &  \\
     \midrule
     {\tt gmr} & .16 & 1.81 & 2.52 & .29 & 2.45 & .27 & .74 & .60 & .36 & .45 & .30 & .48 & .66 & .91 & .50 \\
     {\tt ap} & .34 & .62 & 1.22 & .57 & 1.07 & \underline{.16} & .23 & .48 & \underline{.31} & .39 & .68 & 1.23 & 1.50 & 2.75 & 3.38 \\
     {\tt rlmn} & .26 & .42 & .63 & .32 & .62 & .49 & .73 & 1.17 & .59 & 1.10 & .54 & .79 & .75 & 1.52 & 1.58 \\
     {\tt bcgpr} & \underline{.02} & \underline{.03} & .\underline{03} & \underline{.03} & \underline{.03} & .19
                        & \underline{.17} & \underline{.31} & .32 & \underline{.30} & \underline{.12} & \underline{.11} & \underline{.19} & \underline{.19} & \underline{.16} \\
     {\tt rle\_4} & .41 &  &  & .56 &  & .46 &  &  & .62 &  & .52 &  &  & .69 &  \\
     {\tt rle\_16} & .93 &  &  & 1.23 &  & .97 &  &  & 1.29 &  & 1.08 &  &  & 1.39 &  \\
     {\tt rle\_32} & 1.58 &  &  & 2.07 &  & 1.61 &  &  & 2.10 &  & 1.74 &  &  & 2.28 &  \\
     \bottomrule 	
   \end{tabular}
\end{center}
\label{tbl:queries}
\vspace*{-5ex}
\end{table}

Table~\ref{tbl:queries} shows the query time of individual access, rank, and select
queries. For each type of query, we executed 1,000,000 random queries, reporting
the median time of an individual query achieved over ten non-consecutive
executions. For access query, the best times were reached by {\tt bcgpr}, being 10
times faster. The explanation is that while {\tt bcgpr} performs an
access to a plain sequence during the access query, the other structures
need to perform an access over a succinct representation of a sequence. Notice
that our structure has similar query time than the other structures. For rank
query, the best times were reached by {\tt bcgpr} for large alphabets and by
{\tt ap} for small alphabets. For small alphabets our structure is at most $5$
times slower than the best competitor, using samplings steps of $4$. For large
alphabets, the difference increases to at most $9$ times. For the
case of select queries, the best times were reached by {\tt bcgpr}. In the worst
result, dataset {\tt wl\_2B}, our structure is 12 times slower. In the best
result, dataset {\tt wiki}, our structure is 3.4 times slower. In general, for larger sampling
steps, the query time of our structure increases. However, for larger sampling
steps the size of our structure decreases. Thus, for large enough sequences with
runs, our structure could maintain the sequences in main memory, meanwhile the
other structure should access the disk, increasing their query time.

According to our experimental study, the best trade-offs of our structure are
reached with the underlying structure {\tt ap} and samplings steps of 16 or 32,
if our focus is space, or with {\tt gmr} and samplings of 4 or 16, if our focus
is query time.

%\section{Conclusion}
%...
%Better query time for $r = \Omega(\frac{n}{\mathop{\mathrm{polylog}}(n)})$?
%Better $O(r)$ bits of redundancy? or even $o(r)$?
%
%\myparagraph{Acknowledgement.} ...

%\iffalse
%\makeatletter
%\renewenvironment{thebibliography}[1]
%     {\renewcommand{\baselinestretch}{\smallstretch}\footnotesize
%      \setlength{\itemsep}{-0.155cm plus 0.3cm}
%      \list{\@biblabel{\@arabic\c@enumiv}}%
%           {\settowidth\labelwidth{\@biblabel{#1}}%
%            \leftmargin\labelwidth
%            \advance\leftmargin\labelsep
%            \@openbib@code
%            \usecounter{enumiv}%
%            \let\p@enumiv\@empty
%            \renewcommand\theenumiv{\@arabic\c@enumiv}}%
%      \sloppy\clubpenalty4000\widowpenalty4000%
%      \sfcode`\.\@m}
%     {\def\@noitemerr
%       {\@\LaTeX@warning{Empty `thebibliography' environment}}%
%      \endlist}
%\makeatother
%\fi

%\vspace{-0.2cm}
%\bibliographystyle{IEEEtranS}
%\bibliography{refs}

\addcontentsline{toc}{chapter}{References}
\setlength{\bibsep}{-1pt}
\bibliographystyle{abbrvnat}
\small
\bibliography{refs}

\begin{thebibliography}{19}
\providecommand{\natexlab}[1]{#1}
\providecommand{\url}[1]{\texttt{#1}}
\expandafter\ifx\csname urlstyle\endcsname\relax
  \providecommand{\doi}[1]{doi: #1}\else
  \providecommand{\doi}{doi: \begingroup \urlstyle{rm}\Url}\fi

\bibitem[Barbay et~al.(2014)Barbay, Claude, Gagie, Navarro, and
  Nekrich]{Barbay2014}
J.~Barbay, F.~Claude, T.~Gagie, G.~Navarro, and Y.~Nekrich.
\newblock Efficient fully-compressed sequence representations.
\newblock \emph{Algorithmica}, 69\penalty0 (1):\penalty0 232--268, 2014.

\bibitem[Belazzougui and Navarro(2015)]{BelazzouguiNavarro}
D.~Belazzougui and G.~Navarro.
\newblock Optimal lower and upper bounds for representing sequences.
\newblock \emph{ACM Transactions on Algorithms}, 11\penalty0 (4):\penalty0
  1--21, 2015.

\bibitem[Belazzougui et~al.(2015)Belazzougui, Cunial, Gagie, Prezza, and
  Raffinot]{Belazzougui2015}
D.~Belazzougui, F.~Cunial, T.~Gagie, N.~Prezza, and M.~Raffinot.
\newblock Composite repetition-aware data structures.
\newblock In \emph{Proc. CPM}, pages 26--39. Springer, 2015.

\bibitem[Belazzougui et~al.(2017)Belazzougui, Cunial, Gagie, Prezza, and
  Raffinot]{Belazzougui2017}
D.~Belazzougui, F.~Cunial, T.~Gagie, N.~Prezza, and M.~Raffinot.
\newblock Flexible indexing of repetitive collections.
\newblock In \emph{Proc. CiE}, pages 162--174. Springer, 2017.

\bibitem[Burrows and Wheeler(1994)]{BurrowsWheeler}
M.~Burrows and D.~J. Wheeler.
\newblock {A block-sorting lossless data compression algorithm}.
\newblock Technical Report 124, Digital Equipment Corporation, Palo Alto,
  California, 1994.

\bibitem[Elias(1974)]{Elias}
P.~Elias.
\newblock Efficient storage and retrieval by content and address of static
  files.
\newblock \emph{Journal of the ACM}, 21\penalty0 (2):\penalty0 246--260, 1974.

\bibitem[Fano(1971)]{Fano}
R.~M. Fano.
\newblock On the number of bits required to implement an associative memory.
\newblock Technical Report~61, Computer Structures Group, MIT, Cambridge, MA,
  1971.

\bibitem[Ferragina and Manzini(2000)]{FerraginaManzini}
P.~Ferragina and G.~Manzini.
\newblock Opportunistic data structures with applications.
\newblock In \emph{Proc. FOCS}, pages 390--398. IEEE, 2000.

\bibitem[Gagie et~al.(2018)Gagie, Navarro, and Prezza]{GagieNavarroPrezza}
T.~Gagie, G.~Navarro, and N.~Prezza.
\newblock Optimal-time text indexing in {BWT}-runs bounded space.
\newblock In \emph{Proc. SODA}, pages 1459--1477. SIAM, 2018.

\bibitem[Gog et~al.(2014)Gog, Beller, Moffat, and Petri]{gbmp2014sea}
S.~Gog, T.~Beller, A.~Moffat, and M.~Petri.
\newblock From theory to practice: Plug and play with succinct data structures.
\newblock In \emph{Proc. SEA}, pages 326--337. Springer, 2014.

\bibitem[Golynski et~al.(2006)Golynski, Munro, and Rao]{GolynskiMunroRao}
A.~Golynski, J.~I. Munro, and S.~S. Rao.
\newblock Rank/select operations on large alphabets: a tool for text indexing.
\newblock In \emph{Proc. SODA}, pages 368--373. SIAM, 2006.

\bibitem[Grossi et~al.(2003)Grossi, Gupta, and Vitter]{ggv}
R.~Grossi, A.~Gupta, and J.~S. Vitter.
\newblock High-order entropy-compressed text indexes.
\newblock In \emph{Proc. SODA}, pages 841--850. {ACM/SIAM}, 2003.

\bibitem[Jacobson(1989)]{Jacobson}
G.~Jacobson.
\newblock Space-efficient static trees and graphs.
\newblock In \emph{Proc. FOCS}, pages 549--554. IEEE, 1989.

\bibitem[M\"{a}kinen and Navarro(2005)]{Makinen:2005:SSA:1195881.1195885}
V.~M\"{a}kinen and G.~Navarro.
\newblock Succinct suffix arrays based on run-length encoding.
\newblock \emph{Nordic Journal of Computing}, 12\penalty0 (1):\penalty0 40--66,
  2005.

\bibitem[M{\"a}kinen and Navarro(2007)]{MakinenNavarro}
V.~M{\"a}kinen and G.~Navarro.
\newblock Compressed full-text indexes.
\newblock \emph{ACM Computing Surveys}, 39\penalty0 (1):\penalty0 2--62, 2007.

\bibitem[M{\"a}kinen et~al.(2010)M{\"a}kinen, Navarro, Sir{\'e}n, and
  V{\"a}lim{\"a}ki]{MakinenNavarroSirenValimaki}
V.~M{\"a}kinen, G.~Navarro, J.~Sir{\'e}n, and N.~V{\"a}lim{\"a}ki.
\newblock Storage and retrieval of highly repetitive sequence collections.
\newblock \emph{Journal of Computational Biology}, 17\penalty0 (3):\penalty0
  281--308, 2010.

\bibitem[P{\u{a}}tra{\c{s}}cu and Thorup(2006)]{PatrascuThorup}
M.~P{\u{a}}tra{\c{s}}cu and M.~Thorup.
\newblock Time-space trade-offs for predecessor search.
\newblock In \emph{Proc. STOC}, pages 232--240. ACM, 2006.

\bibitem[Raman et~al.(2007)Raman, Raman, and Rao]{RamanRamanRao}
R.~Raman, V.~Raman, and S.~S. Rao.
\newblock Succinct indexable dictionaries with applications to encoding k-ary
  trees, prefix sums and multisets.
\newblock \emph{ACM Transactions on Algorithms}, 3\penalty0 (4):\penalty0
  43:1--43:25, 2007.

\bibitem[Ziv and Lempel(1977)]{LZ77}
J.~Ziv and A.~Lempel.
\newblock A universal algorithm for sequential data compression.
\newblock \emph{{IEEE} Transactions on Information Theory}, 23\penalty0
  (3):\penalty0 337--343, 1977.

\end{thebibliography}

\end{document}